\documentclass[11pt]{article}

\usepackage[a4paper,margin=2.5cm]{geometry}
\usepackage[T1]{fontenc}
\usepackage[utf8]{inputenc}
\usepackage{mathtools}
\usepackage{amsmath}
\usepackage{amssymb}
\usepackage{amsfonts}
\usepackage{amsthm}
\usepackage{booktabs}
\usepackage{array}
\usepackage{graphicx}
\usepackage{float}
\usepackage{bm}
\usepackage{microtype}
\usepackage[numbers,sort&compress]{natbib}
\usepackage{authblk}
\usepackage[hidelinks]{hyperref}
\newcommand{\doi}[1]{\url{https://doi.org/#1}}
\hypersetup{
  pdftitle={Mutually unbiased bases as extremal probes of isotropic random Hamiltonians},
  pdfauthor={Emilio Semre and Steven Frankel},
  pdfkeywords={mutually unbiased bases, projective designs, random Hamiltonians, Gaussian processes, stochastic order}
}

\theoremstyle{plain}
\newtheorem{theorem}{Theorem}
\newtheorem{lemma}[theorem]{Lemma}
\newtheorem{proposition}[theorem]{Proposition}
\newtheorem{corollary}[theorem]{Corollary}

\theoremstyle{definition}

\theoremstyle{remark}

\newcommand{\C}{\mathbb C}
\newcommand{\R}{\mathbb R}
\newcommand{\E}{\mathbb E}
\newcommand{\Prob}{\mathbb P}
\newcommand{\Id}{I}
\newcommand{\HS}{\mathrm{HS}}
\newcommand{\dd}{\,\mathrm d}
\newcommand{\ket}[1]{\lvert #1\rangle}
\newcommand{\bra}[1]{\langle #1\rvert}
\newcommand{\ketbra}[2]{\lvert #1\rangle\!\langle #2\rvert}
\newcommand{\norm}[1]{\left\lVert #1\right\rVert}
\newcommand{\abs}[1]{\left\lvert #1\right\rvert}
\DeclareMathOperator{\Tr}{Tr}
\DeclareMathOperator{\Var}{Var}
\DeclareMathOperator{\Cov}{Cov}
\DeclareMathOperator{\conv}{conv}

\newenvironment{ruledtabular}{\begin{center}}{\end{center}}

\newcommand{\paperfig}[2]{%
  \IfFileExists{#1}{\includegraphics[width=#2]{#1}}{%
    \fbox{\parbox[c][0.20\textheight][c]{\dimexpr#2-2\fboxsep-2\fboxrule\relax}{\centering
      Missing figure file:\\[2pt]\texttt{\detokenize{#1}}}}%
  }%
}

\begin{document}

\title{Mutually unbiased bases as extremal probes of isotropic random Hamiltonians}

\author[1]{Emilio Semre}
\author[2]{Steven Frankel}
\affil[1]{Department of Computer Science, Technion--Israel Institute of Technology, Haifa, Israel}
\affil[2]{Faculty of Mechanical Engineering, Technion--Israel Institute of Technology, Haifa, Israel}
\date{}

\maketitle

\begin{center}
\small Corresponding author: \href{mailto:abed.semre@campus.technion.ac.il}{abed.semre@campus.technion.ac.il}
\end{center}

\begin{abstract}
Mutually unbiased bases (MUBs) are central finite structures in quantum
information. We ask whether complete MUB systems have an extremal probing
property beyond their projective \(2\)-design identities. For an isotropic
Gaussian traceless Hamiltonian, we prove that, among labeled unions of
\(d+1\) orthonormal bases, a complete MUB union has the stochastically largest
sampled maximum. Each basis induces the same regular-simplex Gaussian block;
mutual unbiasedness eliminates cross-block covariance, joint Gaussianity
yields independence, and a centered-convex Gaussian correlation inequality
makes this independent coupling extremal. We also derive a radial-mixture
extension and prove exact MUB-family collapse for fully matched diagonal-cost
constructions.
\end{abstract}


\section{Introduction}
\label{sec:introduction}

Mutually unbiased bases (MUBs) encode an extremal form of complementarity in a
finite-dimensional quantum system. Two orthonormal bases
\(\mathcal B_a=\{\psi_{a,i}\}_{i=1}^d\) and
\(\mathcal B_b=\{\psi_{b,j}\}_{j=1}^d\) of \(\C^d\) are mutually unbiased if
\begin{equation}
  \abs{\langle\psi_{a,i}\vert\psi_{b,j}\rangle}^2=\frac1d
  \qquad (a\ne b)
  \label{eq:mub-definition-intro}
\end{equation}
for all \(i,j\). A complete system consists of \(d+1\) pairwise mutually
unbiased bases. Such systems exist in every prime-power dimension, including
all qubit-register dimensions \(d=2^n\), while existence in general composite
dimensions remains unresolved \cite{WoottersFields1989,BandyopadhyayEtAl2002}.
The resulting measurement complementarity has been used in state
determination and tomography, cryptographic protocols, and entanglement
detection \cite{WoottersFields1989,SpenglerEtAl2012MUBEntanglement}. Complete
MUB systems also have a distinguished design-theoretic property: their
\(d(d+1)\) projective states form a complex projective state \(2\)-design
\cite{KlappeneckerRoetteler2005,RoyScott2007}.

These facts concern measurement statistics and low-order moments. They leave
open a different geometric question. Suppose that a finite set of pure states
is used to probe a random quantum observable through its expectation values.
Which union of orthonormal bases is most likely to contain a state with an
extreme expectation value? The answer cannot in general be read from the
one-state marginal distribution or from a projective-design identity: the
maximum of a finite random field depends on the full joint dependence between
all sampled values.

We study this question for an isotropic Gaussian random traceless Hamiltonian.
Writing
\begin{equation}
  Q_\psi=\ketbra\psi\psi-\frac{I}{d},
\end{equation}
the sampled value at \(\ket\psi\) is
\(X_\psi=\Tr(HQ_\psi)\). Isotropy gives the covariance kernel
\begin{equation}
  \Cov(X_\psi,X_\phi)
  =\abs{\langle\psi\vert\phi\rangle}^2-\frac1d.
  \label{eq:intro-covariance}
\end{equation}
For a labeled state collection \(\mathcal S\), let
\(M_{\mathcal S}=\max_{\psi\in\mathcal S}X_\psi\). Our comparison class is the
set of labeled unions of \(d+1\) orthonormal bases, with repeated states
allowed through the basis labels.

The main theorem gives a complete answer in every dimension admitting a
complete MUB system. If \(\mathcal M\) is the labeled union of a complete MUB
system and \(\mathcal S\) is any labeled union of \(d+1\) orthonormal bases,
then
\begin{equation}
  \Prob(M_{\mathcal S}\le t)
  \ge \Prob(M_{\mathcal M}\le t)
  \qquad\text{for every }t\in\R.
  \label{eq:intro-stochastic-order}
\end{equation}
Thus \(M_{\mathcal M}\) is largest in the usual stochastic order, not only in
expectation. In particular, the complete MUB union maximizes the Gaussian
width \(\E M_{\mathcal S}\) within the stated comparison class.

The proof exposes the relevant quantum geometry. After the centered-projector
embedding, the \(d\) states of any orthonormal basis are the vertices of a
regular simplex in the real Hilbert space of traceless Hermitian operators.
Accordingly, every basis induces the same regular-simplex Gaussian block.
For two bases, Eq.~\eqref{eq:intro-covariance} shows that their cross-block
covariance is their transition-probability matrix minus the uniform matrix.
Mutual unbiasedness therefore sets the entire cross-block covariance to zero;
joint Gaussianity then makes the blocks independent. Schematically,
\begin{equation}
  \begin{aligned}
    \text{orthonormal basis}
    &\longrightarrow \text{regular-simplex Gaussian block},\\
    \text{mutual unbiasedness}
    &\longrightarrow \text{independent blocks}.
  \end{aligned}
\end{equation}
The remaining step is not a generic independence heuristic. The sublevel sets
of the simplex support function are convex and have zero Gaussian barycenter.
A finite-intersection centered-convex Gaussian correlation inequality
\cite{NakamuraTsuji2025CenteredGCI} then shows that every admissible dependence
places at least as much probability below a common threshold as the independent
coupling. This yields Eq.~\eqref{eq:intro-stochastic-order}.

The result is not implied by the projective \(2\)-design property. A state
\(2\)-design reproduces specified Haar moments under uniform sampling from the
state set. By contrast, Eq.~\eqref{eq:intro-stochastic-order} is a nonlinear
extreme-value statement for a finite jointly Gaussian process and depends on
its complete covariance matrix. Complete MUBs are extremal here because their
basis blocks are exactly decorrelated, not merely because their averaged
second moments match Haar measure.

The isotropic ensemble is a unitarily invariant reference model on the full
traceless-operator space. Local, sparse, diagonal, chemistry, QUBO, QAOA, and
quantum random access optimization (QRAO) Hamiltonians occupy structured,
anisotropic operator subspaces. The theorem therefore does not transfer to
those settings without a separate covariance analysis. We make this boundary
exact for a broad diagonal construction: if a MUB-family circuit has the form
\(D_rH_{\mathrm{Had}}^{\otimes n}\), with \(D_r\) diagonal, and both the
initial state and mixer are transformed by the same circuit, then the rotated
diagonal cost is independent of the family label \(r\). This family-collapse
proposition separates the MUB extremality theorem from claims about generic
variational performance.

MUBs have also been proposed as discrete probes and initial states for
variational quantum algorithms \cite{AlfassiMeiromMor2024DQES}. To delineate
what may occur beyond the isotropic model, we retain two finite-size case
studies from the existing work: an unmatched adaptive MUB-XRot warm start and a
local family search for a non-diagonal QRAO relaxation. They are secondary
illustrations. Their gains are accompanied by unequal parameter counts,
substantial classical-search overhead, and, in the QRAO study, a dominant
classical-prescreen contribution. They provide neither a consequence of the
main theorem nor an asymptotic advantage claim.

The paper is organized as follows. Section~\ref{sec:geometry} develops the
centered-projector and simplex-block representation. Section~\ref{sec:extremality}
proves stochastic extremality and its immediate consequences.
Section~\ref{sec:application-boundary} treats structured Hamiltonians and the
diagonal family-collapse theorem. Section~\ref{sec:variational-illustrations}
summarizes the finite-size boundary tests; their complete protocols and
resource accounting are moved to the appendices. Section~\ref{sec:discussion}
places the result relative to projective designs, structured ensembles, and
variational trainability.

\section{Mutually unbiased bases and the induced Gaussian process}
\label{sec:geometry}

\subsection{Centered projectors and isotropic random observables}

Fix \(d\ge2\). Let \(\mathfrak h_0(d)\) be the real Hilbert space of traceless Hermitian
\(d\times d\) matrices with Hilbert--Schmidt inner product
\begin{equation}
  \langle A,B\rangle_{\HS}=\Tr(AB).
\end{equation}
For a normalized pure state \(\ket\psi\), define
\begin{equation}
  P_\psi=\ketbra\psi\psi,
  \qquad
  Q_\psi=P_\psi-\frac{\Id}{d}.
  \label{eq:centered-projector}
\end{equation}
The identity component of a Hamiltonian shifts every pure-state energy by the
same amount. If
\begin{equation}
  H=H_0+\frac{\Tr H}{d}\Id,
  \qquad H_0\in\mathfrak h_0(d),
\end{equation}
then maximizers, minimizers, and differences between finite state ensembles are
unchanged by replacing \(H\) with \(H_0\). The traceless formulation is
therefore a centering convention rather than a restriction on the physical
spectrum.

A standard isotropic Gaussian random observable (or Hamiltonian) is a centered Gaussian random element
\(H\in\mathfrak h_0(d)\) satisfying
\begin{equation}
  \Cov\!\left(\Tr(HA),\Tr(HB)\right)=\Tr(AB)
  \qquad A,B\in\mathfrak h_0(d).
  \label{eq:isotropic-definition}
\end{equation}
Equivalently, the coordinates of \(H\) in any Hilbert--Schmidt orthonormal
basis of \(\mathfrak h_0(d)\) are independent standard normal variables. The
induced field
\begin{equation}
  X_\psi=\Tr(HQ_\psi)
\end{equation}
has covariance
\begin{align}
  \Cov(X_\psi,X_\phi)
  &=\Tr(Q_\psi Q_\phi) \notag\\
  &=\abs{\langle\psi\vert\phi\rangle}^2-\frac1d.
  \label{eq:field-covariance}
\end{align}
In particular,
\begin{equation}
  \Var(X_\psi)=1-\frac1d,
  \qquad
  d_H(\psi,\phi)^2
  =2\left(1-\abs{\langle\psi\vert\phi\rangle}^2\right),
  \label{eq:canonical-metric}
\end{equation}
where \(d_H\) is the canonical Gaussian metric.

For a finite labeled list \(\mathcal S\) of pure states, define
\begin{equation}
  M_{\mathcal S}=\max_{\psi\in\mathcal S}X_\psi,
  \qquad
  W(\mathcal S)=\E M_{\mathcal S}.
  \label{eq:width-definition}
\end{equation}
Multiplicity is retained because the comparison class below is a labeled
collection of bases; repeated states do not change the maximum but preserve the
basis-block representation.

The ensemble in Eq.~\eqref{eq:isotropic-definition} is unitarily invariant
in the traceless-operator space. The geometry of \(\mathcal S\) is therefore
the only source of covariance. Locality and a preferred Pauli decomposition are
absent. This invariance permits an exact geometric theorem, while also fixing
its scope.

\subsection{Orthonormal bases as regular-simplex Gaussian blocks}

Let
\(\mathcal B_a=\{\psi_{a,1},\ldots,\psi_{a,d}\}\) be an orthonormal basis and
set \(X_{a,i}=X_{\psi_{a,i}}\). Orthonormality and
Eq.~\eqref{eq:field-covariance} give
\begin{equation}
  \Cov(X_{a,i},X_{a,j})=\delta_{ij}-\frac1d,
  \qquad
  \sum_{i=1}^d X_{a,i}=0.
  \label{eq:basis-covariance}
\end{equation}
Choose vectors \(v_1,\ldots,v_d\in\R^{d-1}\) satisfying
\begin{equation}
  v_i\cdot v_j=\delta_{ij}-\frac1d,
  \qquad
  \sum_{i=1}^d v_i=0.
  \label{eq:simplex-vectors}
\end{equation}
These are the vertices of a centered regular simplex. Let \(V\) be the
\(d\times(d-1)\) matrix whose \(i\)th row is \(v_i^{\mathsf T}\). Then
\begin{equation}
  VV^{\mathsf T}=I_d-\frac1d\bm 1\bm 1^{\mathsf T},
  \qquad
  V^{\mathsf T}V=I_{d-1}.
  \label{eq:V-identities}
\end{equation}

\begin{lemma}[Simplex-block representation]
\label{lem:simplex-block}
Let \(\mathcal S=(\mathcal B_1,\ldots,\mathcal B_m)\) be a labeled union of
orthonormal bases. There exist jointly Gaussian vectors
\((Y_1,\ldots,Y_m)\), each with marginal law \(N(0,I_{d-1})\), such that
\begin{equation}
  X_{a,i}=\langle v_i,Y_a\rangle
  \qquad (a=1,\ldots,m;\ i=1,\ldots,d).
  \label{eq:block-representation}
\end{equation}
For \(a\ne b\), their cross-covariance is
\begin{equation}
  \Cov(Y_a,Y_b)=V^{\mathsf T}K^{ab}V,
  \qquad
  K^{ab}_{ij}=\abs{\langle\psi_{a,i}\vert\psi_{b,j}\rangle}^2-\frac1d.
  \label{eq:cross-block-covariance}
\end{equation}
If the bases are mutually unbiased, then \(Y_1,\ldots,Y_m\) are independent.
\end{lemma}

\begin{proof}
Write \(X_a=(X_{a,1},\ldots,X_{a,d})^{\mathsf T}\). By
Eq.~\eqref{eq:basis-covariance}, \(X_a\in\bm{1}^\perp\) almost surely and its
covariance is the orthogonal projector onto \(\bm{1}^\perp\). Define
\(Y_a=V^{\mathsf T}X_a\). Then \(Y_a\) is standard Gaussian in
\(\R^{d-1}\), while Eqs.~\eqref{eq:V-identities} and
\(X_a\in\bm{1}^\perp\) give \(X_a=VY_a\), proving
Eq.~\eqref{eq:block-representation}. Equation~\eqref{eq:cross-block-covariance}
follows by taking cross-covariances. For mutually unbiased bases,
\(K^{ab}=0\); hence the blocks are uncorrelated. Joint Gaussianity then implies
independence.
\end{proof}

Thus all unions of orthonormal bases have exactly the same marginal block law.
The basis choice enters only through the admissible cross-block covariance
matrices in Eq.~\eqref{eq:cross-block-covariance}. Complete mutual unbiasedness
selects the fully decorrelated point of this covariance class.

\subsection{Projective designs versus extreme-value functionals}

A complete MUB ensemble is a complex projective \(2\)-design
\cite{KlappeneckerRoetteler2005,RoyScott2007}. This implies exact first- and
second-moment identities under uniform sampling from its \(d(d+1)\) states.
The functional in Eq.~\eqref{eq:width-definition}, however, is a maximum of a
joint Gaussian vector. Its distribution depends on the entire covariance
matrix, including the block structure retained by the labeled basis union.
Consequently, projective-design moment matching is useful context but is not a
proof of extreme-value optimality.

\section{Stochastic extremality of complete MUB systems}
\label{sec:extremality}

\subsection{Centered-convex correlation for Gaussian blocks}

We first isolate the probabilistic mechanism. The following proposition is a
direct specialization of the finite-intersection centered-convex Gaussian
correlation theorem of Nakamura and Tsuji
\cite[Theorem~1.5]{NakamuraTsuji2025CenteredGCI}; its novelty here is its
identification with the basis-union problem.

\begin{proposition}[Identical Gaussian blocks with centered sublevel sets]
\label{prop:gaussian-block}
Let \((Y_1,\ldots,Y_m)\) be jointly Gaussian vectors in \(\R^q\), with
\(Y_a\sim N(0,I_q)\) for every \(a\). Let \(K\subseteq\R^q\) be a convex
Borel set with zero Gaussian barycenter,
\begin{equation}
  \int_K y\,\dd\gamma_q(y)=0,
  \label{eq:gaussian-barycenter}
\end{equation}
where \(\gamma_q\) denotes standard Gaussian measure. Then
\begin{equation}
  \Prob(Y_1\in K,\ldots,Y_m\in K)
  \ge \gamma_q(K)^m.
  \label{eq:block-correlation}
\end{equation}
The right-hand side is the simultaneous-membership probability for independent
standard Gaussian blocks.
\end{proposition}

\begin{proof}
There are a standard Gaussian vector \(G\sim N(0,I_s)\) and linear maps
\(B_a:\R^s\to\R^q\) such that \(Y_a=B_aG\) and
\(B_aB_a^{\mathsf T}=I_q\). Define the cylinders
\begin{equation}
  C_a=\{g\in\R^s:B_ag\in K\}.
\end{equation}
They are convex. Gaussian regression gives
\(\E[G\mid Y_a]=B_a^{\mathsf T}Y_a\), and therefore
\begin{align}
  \int_{C_a}g\,\dd\gamma_s(g)
  &=\E\!\left[G\mathbf{1}_{\{Y_a\in K\}}\right] \notag\\
  &=B_a^{\mathsf T}\E\!\left[Y_a\mathbf{1}_{\{Y_a\in K\}}\right]=0.
  \label{eq:cylinder-centering}
\end{align}
The finite-intersection centered-convex Gaussian correlation inequality applied
to \(C_1,\ldots,C_m\) yields
\begin{equation}
  \gamma_s\!\left(\bigcap_{a=1}^m C_a\right)
  \ge\prod_{a=1}^m\gamma_s(C_a)
  =\gamma_q(K)^m,
\end{equation}
which is Eq.~\eqref{eq:block-correlation}.
\end{proof}

The finite-intersection theorem is essential here. Iterating only a two-set
inequality would not suffice because an intersection of centered convex sets
need not retain zero Gaussian barycenter.

\subsection{Centered sublevel sets of the simplex support function}

Define the simplex support function
\begin{equation}
  h(y)=\max_{1\le i\le d}\langle v_i,y\rangle
  \label{eq:simplex-support}
\end{equation}
and its sublevel set
\begin{equation}
  K_t=\{y\in\R^{d-1}:h(y)\le t\}.
  \label{eq:Kt}
\end{equation}
Because \(\sum_i\langle v_i,y\rangle=0\), one has \(h(y)\ge0\) for every
\(y\). Thus \(K_t=\varnothing\) for \(t<0\).

\begin{lemma}[Gaussian centering of simplex sublevel sets]
\label{lem:Kt-centered}
For every \(t\in\R\), \(K_t\) is convex and
\begin{equation}
  \int_{K_t}y\,\dd\gamma_{d-1}(y)=0.
  \label{eq:Kt-barycenter}
\end{equation}
\end{lemma}

\begin{proof}
Convexity follows because \(h\) is a support function. The claim is trivial for
\(t<0\). For \(t\ge0\), \(K_t\) is invariant under every orthogonal symmetry
that permutes the regular-simplex vertices. Standard Gaussian measure has the
same invariance. Its restricted first moment must therefore be fixed by the
standard \((d-1)\)-dimensional representation of the simplex symmetry group.
That representation has no nonzero fixed vector, so the first moment vanishes.
\end{proof}

\subsection{Stochastic extremality theorem}

\begin{theorem}[Stochastic extremality of complete MUB unions]
\label{thm:main}
Assume that \(d+1\) mutually unbiased bases exist in \(\C^d\), and let
\(\mathcal M\) be their labeled union. Let
\(\mathcal S=(\mathcal B_1,\ldots,\mathcal B_{d+1})\) be any labeled union of
\(d+1\) orthonormal bases, with basis membership counted with multiplicity.
Then
\begin{equation}
  \Prob(M_{\mathcal S}\le t)
  \ge
  \Prob(M_{\mathcal M}\le t)
  \qquad\text{for every }t\in\R.
  \label{eq:stochastic-order}
\end{equation}
Equivalently, \(M_{\mathcal S}\) is smaller than \(M_{\mathcal M}\) in the
usual stochastic order. In particular,
\begin{equation}
  W(\mathcal S)\le W(\mathcal M).
  \label{eq:width-order}
\end{equation}
By the symmetry \(H\stackrel{\mathrm d}= -H\),
\begin{equation}
  \E\min_{\psi\in\mathcal M}\Tr(HQ_\psi)
  \le
  \E\min_{\psi\in\mathcal S}\Tr(HQ_\psi).
  \label{eq:min-order}
\end{equation}
\end{theorem}

\begin{proof}
By Lemma~\ref{lem:simplex-block},
\begin{equation}
  M_{\mathcal S}\stackrel{\mathrm d}=
  \max_{1\le a\le d+1}h(Y_a),
\end{equation}
where the \(Y_a\) are jointly Gaussian with standard marginals. Hence, for
\(t\ge0\),
\begin{align}
  \Prob(M_{\mathcal S}\le t)
  &=\Prob(Y_1\in K_t,\ldots,Y_{d+1}\in K_t) \notag\\
  &\ge \gamma_{d-1}(K_t)^{d+1},
  \label{eq:dependent-cdf-bound}
\end{align}
where Proposition~\ref{prop:gaussian-block} and
Lemma~\ref{lem:Kt-centered} were used. For a complete MUB system, the blocks
are independent by Lemma~\ref{lem:simplex-block}; therefore
\begin{equation}
  \Prob(M_{\mathcal M}\le t)=\gamma_{d-1}(K_t)^{d+1}.
  \label{eq:mub-cdf}
\end{equation}
Equations~\eqref{eq:dependent-cdf-bound} and \eqref{eq:mub-cdf} prove
Eq.~\eqref{eq:stochastic-order} for \(t\ge0\). Both probabilities vanish for
\(t<0\) because \(h\ge0\).

The maxima have finite expectations and are nonnegative. Integrating their
upper tails gives
\begin{equation}
  \E M=\int_0^\infty \Prob(M>t)\,\dd t,
\end{equation}
so Eq.~\eqref{eq:stochastic-order} implies Eq.~\eqref{eq:width-order}. Finally,
\begin{equation}
  \E\min_{\psi\in\mathcal S}X_\psi
  =-\E\max_{\psi\in\mathcal S}(-X_\psi)
  =-W(\mathcal S),
\end{equation}
because \((-X_\psi)_\psi\) has the same law as \((X_\psi)_\psi\). This gives
Eq.~\eqref{eq:min-order}.
\end{proof}

The theorem identifies a precise quantum-geometric principle: within the
admissible class of \(d+1\) orthonormal-basis blocks, complete mutual
unbiasedness realizes full Gaussian block decorrelation, and this decorrelation
maximizes the sampled extreme.

\subsection{Gaussian width, radial mixtures, and the extreme-value scale}

\begin{corollary}[Independent radial mixtures]
\label{cor:radial}
Let \(G\) be a standard isotropic Gaussian element of \(\mathfrak h_0(d)\),
let \(R\ge0\) be independent of \(G\), and set \(H=RG\). The stochastic order
in Eq.~\eqref{eq:stochastic-order} remains valid for the corresponding maxima.
If \(\E R<\infty\), then the expectation inequalities
Eqs.~\eqref{eq:width-order} and \eqref{eq:min-order} also remain valid.
\end{corollary}

\begin{proof}
Conditioned on \(R=r\), every sampled energy and its maximum are multiplied by
\(r\). The conditional stochastic order follows from Theorem~\ref{thm:main}
and is preserved after averaging over \(R\). For the expectation,
\begin{equation}
  \E\max_{\psi\in\mathcal S}\Tr(RGQ_\psi)
  =(\E R)W(\mathcal S),
\end{equation}
and similarly for \(\mathcal M\). Symmetry gives the minimum statement.
\end{proof}

This corollary uses only an independent nonnegative common scale. It is not an
extension to arbitrary non-Gaussian or anisotropic Hamiltonian ensembles.

Let \(N=d(d+1)\). The complete MUB ensemble also has the generic finite-Gaussian
extreme-value scale
\begin{equation}
  W(\mathcal M)=\Theta(\sqrt{\log N}),
  \qquad
  \E\min_{\psi\in\mathcal M}X_\psi=-\Theta(\sqrt{\log N}).
  \label{eq:rate-scale}
\end{equation}
Indeed, the universal Gaussian maximal bound gives the upper estimate because
each variance is \(1-1/d\). Distinct MUB states have canonical distance at
least \(1\) by Eq.~\eqref{eq:canonical-metric}, so Sudakov minoration gives the
matching lower estimate \cite{LedouxTalagrand1991,BoucheronLugosiMassart2013}.
Appendix~\ref{app:rate} records the normalization. This is a standard
Gaussian-process rate observation; the new statement is the finite-dimensional
stochastic comparison in Theorem~\ref{thm:main}.

For \(d=2\), the complete MUB gives the six vertices of the regular octahedron
on the Bloch sphere. Appendix~\ref{app:qubit} verifies, using the
Cauchy--Hadwiger mean-width relation and Linhart's sharp edge-curvature bound,
that this octahedron also maximizes \(W\) among arbitrary lists of six pure
qubit states. This low-dimensional strengthening relies on three-dimensional
convex geometry and is not used elsewhere in the paper.

\subsection{Scope of the theorem}
\label{sec:not-say}

Theorem~\ref{thm:main} is exact, but its domain is narrow. It does not imply
optimality among arbitrary \(d(d+1)\)-state ensembles for \(d\ge3\), or for
arbitrary non-Gaussian, local, sparse, QUBO, QRAO, or chemistry Hamiltonian
distributions. It does not imply a better QAOA approximation ratio, a better
decoded combinatorial solution, easier optimization, avoidance of barren
plateaus, lower total runtime, or a quantum computational advantage. Those
questions depend on anisotropy, circuit dynamics, parameterization,
measurement, decoding, and classical search, none of which appears in the
isotropic basis-union theorem.

\section{Boundary for structured quantum Hamiltonians}
\label{sec:application-boundary}
\subsection{From finite probing to variational initialization}

For fixed variational angles \(\theta\), an initialized variational circuit
samples one Hermitian observable:
\begin{equation}
  \bra\psi U(\theta)^\dagger H_CU(\theta)\ket\psi.
  \label{eq:fixed-angle-observable}
\end{equation}
This formal identity is the only universal bridge between finite-set coverage
and initialization. It does not commute with angle optimization: different
initial states generally induce different landscapes, optimizer trajectories,
and resource costs. The isotropic theorem may therefore motivate a probe set,
but it cannot establish that optimizing from that set improves a structured
variational task.

The contrast is particularly sharp for physical and combinatorial
Hamiltonians. Such Hamiltonians typically occupy a sparse or local subspace of
the full traceless-Hermitian space. Their induced state-energy field is
anisotropic, and cross-basis covariance is no longer determined solely by
Eq.~\eqref{eq:mub-definition-intro}. Locality, constraints, and decoding may
favor one MUB family, a product subfamily, or no MUB structure at all.

The next proposition gives an exact structural boundary for a common matched
construction. It shows why nontrivial applications require anisotropy,
noncommutativity, unmatched pre-circuit freedom, or a combination of these
features.

\subsection{Exact family collapse for diagonal matched constructions}
\label{sec:diagonal-collapse}

Consider an \(n\)-qubit MUB-family circuit
\begin{equation}
  C_r=D_rH_{\mathrm{Had}}^{\otimes n},
  \label{eq:Cr-form}
\end{equation}
where \(D_r\) is diagonal in the computational basis. For a computational label \(b\), a fully matched construction prepares
\(C_r\ket b\) and uses the standard computational-basis \(Z\)-reference
mixer \(B_b\), transformed as
\begin{equation}
  B_{r,b}=C_rB_bC_r^\dagger
\end{equation}
whose reference ground state is \(C_r\ket b\). Conjugating the entire
variational circuit by \(C_r\) replaces the cost by \(C_r^\dagger H_CC_r\),
the mixer by \(B_b\), and the initial state by \(\ket b\).

\begin{proposition}[Diagonal-cost family collapse]
\label{prop:diagonal-collapse}
Let \(H_C\) be diagonal in the computational basis and let \(C_r\) have the
form in Eq.~\eqref{eq:Cr-form}. Then
\begin{equation}
  C_r^\dagger H_CC_r
  =H_{\mathrm{Had}}^{\otimes n}H_CH_{\mathrm{Had}}^{\otimes n},
  \label{eq:rotated-cost-collapse}
\end{equation}
which is independent of \(r\). Consequently, every fully matched family has the
same variational landscape as a function of the shared QAOA angles. Moreover,
\begin{equation}
  \bra bC_r^\dagger H_CC_r\ket b=2^{-n}\Tr(H_C)
  \label{eq:prescreen-collapse}
\end{equation}
for every \(b\), so this construction also has no family- or label-dependent
initial-energy prescreen. For \(b=0\), it is unitarily equivalent to ordinary
transverse-\(X\) QAOA.
\end{proposition}

\begin{proof}
Since both \(D_r\) and \(H_C\) are diagonal, \([D_r,H_C]=0\). Hence
\begin{align}
  C_r^\dagger H_CC_r
  &=H_{\mathrm{Had}}^{\otimes n}D_r^\dagger H_CD_r
    H_{\mathrm{Had}}^{\otimes n} \notag\\
  &=H_{\mathrm{Had}}^{\otimes n}H_CH_{\mathrm{Had}}^{\otimes n},
\end{align}
proving Eq.~\eqref{eq:rotated-cost-collapse}. If
\(H_C=\sum_x c_x\ketbra x x\), every computational amplitude of
\(H_{\mathrm{Had}}^{\otimes n}\ket b\) has squared magnitude \(2^{-n}\), so
\begin{equation}
  \bra bH_{\mathrm{Had}}^{\otimes n}H_CH_{\mathrm{Had}}^{\otimes n}\ket b
  =2^{-n}\sum_xc_x.
\end{equation}
This is Eq.~\eqref{eq:prescreen-collapse}. For \(b=0\), a further conjugation
by \(H_{\mathrm{Had}}^{\otimes n}\) maps \(\ket0^{\otimes n}\) to
\(\ket+^{\otimes n}\), a longitudinal \(Z\) reference mixer to the transverse
\(X\) mixer, and the rotated cost back to \(H_C\).
\end{proof}

The proposition rules out a family-selection mechanism for diagonal costs in
this fully matched setting. A family-dependent practical construction must
therefore change at least one ingredient: the pre-circuit need not be matched
to the mixer, the Hamiltonian may be non-diagonal, or the state-selection rule
may use additional classical information.

\section{Finite-size variational illustrations outside the isotropic model}
\label{sec:variational-illustrations}

The following two studies are retained to show how family dependence can
reappear once the assumptions of Theorem~\ref{thm:main} are broken. They are
not tests of the theorem itself. The complete ansatz definitions, benchmark
grids, tables, figures, and resource accounting are given in
Appendices~\ref{app:xrot-full} and \ref{app:qrao-full}.

\subsection{Composite adaptive MUB-XRot warm start}
\label{sec:mub-xrot}

Adaptive MUB-XRot uses an unmatched MUB-family pre-circuit, \(n\) additional
local \(X\)-rotation parameters, ordinary shared-angle QAOA layers, and a local
discrete family search. The baseline optimizes only the \(2p\) shared QAOA
angles. The available benchmark therefore does not isolate MUB geometry from
additional expressivity or optimization budget, and it contains no XRot-only
or equally parameterized generic pre-circuit control.

Across the retained finite-size benchmark, the mean decoded-ratio changes are
heterogeneous: \(+0.0100\) for MaxCut, \(+0.0366\) for weighted MaxCut,
\(+0.3665\) for MIS, and \(+0.2792\) for weighted MIS. The original aggregate
mean \(+0.1616\) also includes a secondary knapsack implementation with an
unresolved feasibility/slack-consistency issue and treats repeated circuit
depths of the same instance as separate descriptive rows. The overall median
runtime ratio is approximately \(46\times\). These observations support only
the composite warm-start procedure under a substantially larger classical and
variational budget.

\subsection{Local MUB-family search in non-diagonal QRAO}
\label{sec:qrao}

The QRAO study uses a non-diagonal \((3,1)\)-QRAC relaxation, so the diagonal
collapse in Proposition~\ref{prop:diagonal-collapse} does not apply. Its
bit-flip multi-start method attains mean relaxed approximation ratio
\(\alpha_r=0.9211\), a mean change \(+0.0608\) relative to the
\(X\)-variational baseline, and decoded solved rates of \(61.9\%\) versus
\(46.9\%\). The improvement is not uniform: the method loses in relaxed ratio
on 72 of 360 graph--depth cells.

The attribution decomposition is essential:
\begin{equation}
  0.0608=0.0462+0.0039+0.0107,
  \label{eq:qrao-decomposition-main}
\end{equation}
where the three terms are respectively due to a classical label prescreen,
comparison of two initial family poles, and the Hamming-neighbor family walk.
The median runtime ratio relative to the \(X\)-variational baseline is
approximately \(9.2\times\), and no asymptotic bound is proved for the family
search. Thus the result concerns a composite classical--variational procedure
for a structured relaxation, not an isolated MUB-family effect or a scaling
advantage.

\section{Discussion}
\label{sec:discussion}

\subsection{Extremality beyond projective-design identities}

The comparison in Theorem~\ref{thm:main} holds because the basis constraint
fixes every marginal block law. Each orthonormal basis is represented by the
same centered regular simplex, while the transition probabilities between
bases determine the admissible Gaussian coupling of those simplex blocks. A
complete MUB system sets all cross-block covariances to zero. The
centered-convex correlation inequality then orders the entire distribution of
the maximum, rather than only its mean or variance.

This mechanism explains why the theorem is not a reformulation of the
projective \(2\)-design property. Design identities average tensor powers of
individual projectors. The maximum \(M_{\mathcal S}\) instead depends on the
joint covariance matrix of all labeled probe values. In the Gaussian model that
matrix determines the complete law, and the MUB transition probabilities
select the independent-block coupling that is extremal for the relevant
simplex sublevel sets.

\subsection{Structured and anisotropic Hamiltonian ensembles}

Locality, sparse Pauli support, constraints, and encodings introduce preferred
directions in operator space. For such ensembles the covariance of
\(\Tr(HQ_\psi)\) is no longer a universal function only of
\(\abs{\langle\psi\vert\phi\rangle}^2\), and mutual unbiasedness need not
remove cross-family dependence. Proposition~\ref{prop:diagonal-collapse}
shows that this difference can be qualitative: in a fully matched diagonal
construction, all MUB-family dependence disappears exactly. Conversely, the
finite-size QRAO example exhibits family dependence only after a non-diagonal
anisotropic encoding and classical prescreen are introduced.

The appropriate extension of the present theorem is therefore not a generic
claim that MUBs improve structured variational algorithms. It is to identify
anisotropic quantum-observable ensembles whose induced block covariances obey a
comparison principle strong enough to order their maxima. Companion work on a
product MUB family in QRAO studies one such structured setting directly rather
than deriving it from isotropy \cite{SemreFrankel2026RandSEMI}.

\subsection{Projective state designs and variational trainability}
\label{sec:trainability}

A complete MUB ensemble is a projective state \(2\)-design, not a unitary
\(2\)-design. Barren-plateau results concern the distribution of parameterized
unitaries or circuit subblocks, together with the cost observable, depth, and
parameter distribution \cite{McCleanEtAl2018BarrenPlateaus,CerezoEtAl2021CostDependent}.
The MUB state-design property therefore implies neither the presence nor the
absence of a barren plateau. The finite-size studies in
Section~\ref{sec:variational-illustrations} do not measure asymptotic gradient
variance, shot complexity, or noise scaling and provide no trainability result.

\section{Conclusion}
\label{sec:conclusion}

Among labeled unions of \(d+1\) orthonormal bases in \(\C^d\), whenever a
complete MUB system exists, its union is a stochastic-extremal finite probe of
a standard isotropic Gaussian traceless Hamiltonian. Every basis induces the
same regular-simplex Gaussian block; mutual unbiasedness eliminates every
cross-block covariance; joint Gaussianity turns this decorrelation into
independence; and a centered-convex Gaussian correlation inequality shows that
the resulting independent blocks produce the stochastically largest maximum.
The expected Gaussian-width inequality follows immediately.

This conclusion is not contained in the known projective \(2\)-design identity
for complete MUBs. The design property fixes low-order averaged moments,
whereas the theorem orders a nonlinear extreme-value functional through the
full covariance structure of the induced finite Gaussian process. Its domain is
correspondingly precise: unions of \(d+1\) orthonormal bases and an isotropic
Gaussian traceless-Hamiltonian field.

Structured Hamiltonians require a different analysis. The exact diagonal
family-collapse proposition shows that a fully matched MUB construction can
lose all family dependence, while the retained variational examples obtain
their effects from unmatched circuits, anisotropy, decoding, prescreening, and
additional optimization resources. The most direct theoretical continuation is
to determine which structured or anisotropic quantum-Hamiltonian ensembles
admit an analogous covariance-extremality principle and whether that principle
can be related rigorously to operational variational objectives.

\clearpage
\appendix

\section{Normalization of the Gaussian extreme-value scale}
\label{app:rate}

Let \(\mathcal S\) be any list of \(N\) pure states. Since
\(\Var(X_\psi)=1-1/d\le1\), the standard exponential-moment bound gives
\begin{equation}
  W(\mathcal S)=\E\max_{\psi\in\mathcal S}X_\psi
  \le\sqrt{2\log N}.
  \label{eq:maximal-upper}
\end{equation}
For a complete MUB ensemble, distinct states satisfy either
\(\abs{\langle\psi\vert\phi\rangle}^2=0\) within a basis or
\(1/d\) across bases. Equation~\eqref{eq:canonical-metric} therefore gives
\begin{equation}
  d_H(\psi,\phi)^2\in\left\{2,\,2\left(1-\frac1d\right)\right\},
\end{equation}
so the minimum pairwise distance is at least \(1\) for \(d\ge2\). Sudakov
minoration then yields a universal \(c>0\) such that
\begin{equation}
  W(\mathcal M)\ge c\sqrt{\log N}.
  \label{eq:sudakov-lower}
\end{equation}
Together, Eqs.~\eqref{eq:maximal-upper} and \eqref{eq:sudakov-lower} prove
Eq.~\eqref{eq:rate-scale}. No claim is made about the sharp constants.

\section{Unrestricted six-state extremality for a qubit}
\label{app:qubit}

For \(d=2\), write
\begin{equation}
  Q_r=\frac12 r\cdot\bm\sigma,
  \qquad r\in S^2,
\end{equation}
and expand a standard Gaussian traceless Hamiltonian as
\begin{equation}
  H=\frac1{\sqrt2}g\cdot\bm\sigma,
  \qquad g\sim N(0,I_3).
\end{equation}
Then
\begin{equation}
  \Tr(HQ_r)=\frac1{\sqrt2}g\cdot r.
  \label{eq:bloch-gaussian}
\end{equation}
For six states with Bloch vectors \(r_1,\ldots,r_6\), let
\(P=\conv\{r_1,\ldots,r_6\}\). Writing \(g=Ru\), where \(u\) is uniform on
\(S^2\) and independent of \(R=\norm g\), gives
\begin{equation}
  W=\frac{\E R}{\sqrt2}\int_{S^2}h_P(u)\,\dd\sigma(u),
  \label{eq:width-mean-support}
\end{equation}
with \(\sigma\) normalized surface measure. Thus maximizing \(W\) is equivalent
to maximizing the mean width of a convex polytope with at most six vertices and
circumradius at most one.

For a full-dimensional convex polyhedron, the Cauchy--Hadwiger relation
identifies mean width with the total edge curvature up to a universal positive
factor \cite{Schneider2014ConvexBodies}. Let \(v,e,f\) denote the numbers of
vertices, edges, and faces, and let \(R_P\) be the circumradius. Linhart's sharp
inequality \cite{Linhart1979EdgeCurvature} bounds the total edge curvature
\(M(P)\) by
\begin{align}
 \frac{M(P)}{R_P}
 &\le 2e\sin\!\left(\frac{\pi f}{2e}\right)
 \left[1-\cot^2\!\left(\frac{\pi f}{2e}\right)
          \cot^2\!\left(\frac{\pi v}{2e}\right)\right]^{1/2}
 \notag\\
 &\quad\times
 \arccos\!\left[
   \cos\!\left(\frac{\pi v}{2e}\right)
   \csc\!\left(\frac{\pi f}{2e}\right)
 \right],
 \label{eq:linhart}
\end{align}
with equality precisely for the Platonic solids.

Euler's relation and the planar graph inequalities leave the following
full-dimensional combinatorial possibilities for \(v\le6\):
\begin{table}[H]
\caption{Linhart upper bound in Eq.~\eqref{eq:linhart} at circumradius one for
all full-dimensional combinatorial triples with at most six vertices.}
\label{tab:linhart-check}
\begin{ruledtabular}
\begin{tabular}{cccc}
\(v\) & \(e\) & \(f\) & Upper bound \\
\hline
4 & 6 & 4 & 9.360 \\
5 & 8 & 5 & 9.988 \\
5 & 9 & 6 & 10.015 \\
6 & 9 & 5 & 10.370 \\
6 & 10 & 6 & 10.410 \\
6 & 11 & 7 & 10.432 \\
6 & 12 & 8 & 10.445 \\
\end{tabular}
\end{ruledtabular}
\end{table}
The largest bound occurs for \((v,e,f)=(6,12,8)\), and equality requires the
regular octahedron. Lower-dimensional hulls follow by a perturbation within
\(S^2\) and continuity of support functions in the Hausdorff metric. Therefore
the six Bloch vectors \(\{\pm e_1,\pm e_2,\pm e_3\}\), which are exactly the
three qubit MUBs, maximize Eq.~\eqref{eq:width-mean-support} among arbitrary
six-state lists. Repeated or coincident states are included by the same
continuity argument.

\clearpage
\section{Adaptive MUB-XRot benchmark and resource accounting}
\label{app:xrot-full}

\subsection{Ansatz and resource asymmetry}

The adaptive MUB-XRot study uses an unmatched warm-start ansatz rather than the
fully matched construction of Sec.~\ref{sec:diagonal-collapse}. Standard
shared-angle depth-\(p\) QAOA begins from \(\ket+^{\otimes n}\) and optimizes
\(2p\) continuous angles. Adaptive MUB-XRot prepares
\begin{equation}
  \ket{\phi_{j,\mu}}
  =F_j\left(\prod_{i=1}^n R_X(\mu_i)\right)\ket0^{\otimes n},
  \label{eq:mub-xrot-state}
\end{equation}
then applies ordinary QAOA cost and transverse-mixer layers. Here \(F_j\) is a
noncomputational MUB-family circuit implemented using the chosen complete-MUB
construction \cite{WangWu2024MUBCircuits}. The mixer is not chosen to make
\(\ket{\phi_{j,\mu}}\) its ground state.

The comparison is not parameter matched. Adaptive MUB-XRot optimizes
\(n+2p\) continuous parameters and a discrete family variable, whereas the
baseline optimizes \(2p\) shared QAOA angles. The available benchmark contains
no XRot-only, fixed-MUB-plus-XRot, equally parameterized generic local-rotation,
or comparable random structured pre-circuit control. An observed improvement
therefore belongs to the composite adaptive MUB-XRot ansatz; it cannot be
attributed specifically to MUB geometry.

\begin{figure}[H]
  \centering
  \paperfig{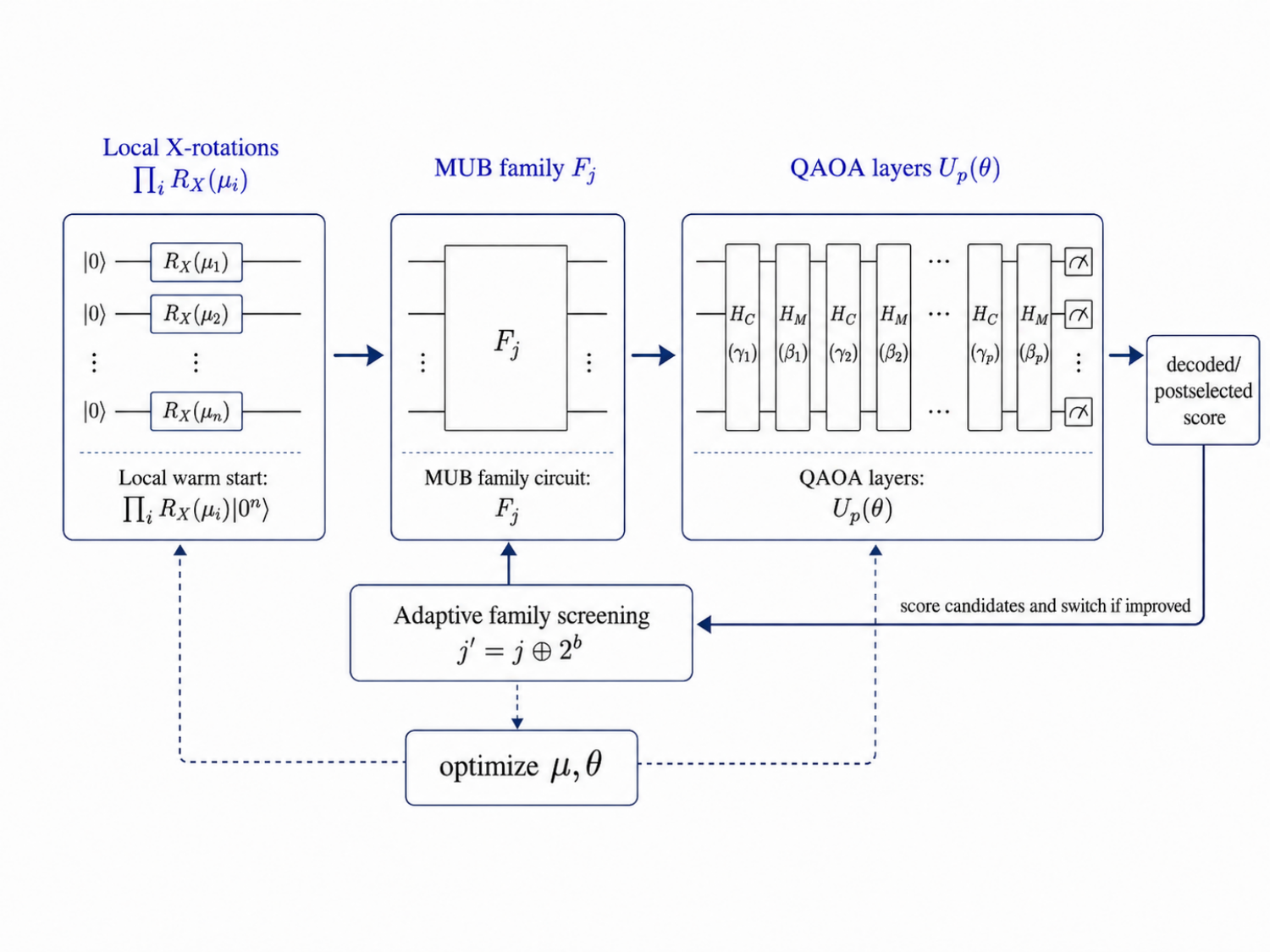}{0.72\textwidth}
  \caption{Composite adaptive MUB-XRot warm start. Local rotations and a selected
  MUB-family circuit prepare Eq.~\eqref{eq:mub-xrot-state}; ordinary shared-angle
  QAOA layers then follow. Hamming-neighbor families are screened and selected
  during the classical outer loop. The method adds \(n\) continuous parameters
  and a discrete family search relative to standard QAOA, so the comparison is
  not an equal-resource test of MUB structure alone.}
  \label{fig:mub-xrot-method}
\end{figure}

The family index starts from a fixed noncomputational family. The algorithm
screens the current family and valid bit-flip neighbors \(j\oplus2^b\), ranks
candidates using postselected expected ratio, decoded ratio, and energy, and
reoptimizes the leading candidates under an equal small budget. A switch is
accepted only when the best candidate improves the screening objective by at
least \(10^{-4}\). Continuous optimization uses L-BFGS-B
\cite{ByrdLuNocedalZhu1995}. This procedure is a local discrete search over
families, not an exhaustive scan of all MUB states.

\subsection{Metrics and descriptive benchmark}

The primary reported metric is the decoded solution ratio: the objective value
of the highest-probability feasible decoded bit string divided by the
brute-force optimum. Infeasible strings receive zero decoded score for the
constrained problems. A tie is defined computationally by
\(\abs{\Delta}\le10^{-9}\), where \(\Delta\) is the paired difference in
decoded ratio. This threshold identifies numerical equality; it is not a
statistical equivalence margin.

The source benchmark uses
\begin{equation}
  n\in\{8,10,12,14\},\qquad
  p\in\{1,2,3\},\qquad
  25\text{ seeds},
  \label{eq:xrot-grid}
\end{equation}
for MaxCut, weighted MaxCut, maximum independent set (MIS), weighted MIS, and
knapsack. There are 500 underlying instances and 1500 paired instance--depth
rows. The same instance appears at multiple depths, so the 1500 rows are not
independent experimental units. In the absence of raw clustered data, all
summaries below are descriptive; no row-wise significance test is reported.

The mean decoded-ratio changes are strongly heterogeneous:
\begin{table}[H]
\caption{Mean paired change in decoded solution ratio for the composite
adaptive MUB-XRot ansatz relative to standard shared-angle QAOA. The knapsack
entry is secondary because the current implementation has an unresolved
feasibility/slack-consistency issue.}
\label{tab:xrot-by-problem}
\begin{ruledtabular}
\begin{tabular}{lc}
Problem family & Mean change \\
\hline
MaxCut & \(+0.0100\) \\
Weighted MaxCut & \(+0.0366\) \\
MIS & \(+0.3665\) \\
Weighted MIS & \(+0.2792\) \\
Knapsack (secondary) & \(+0.1158\) \\
\end{tabular}
\end{ruledtabular}
\end{table}
The aggregate is therefore driven primarily by the two MIS families, while
unweighted MaxCut has a ceiling effect: the standard baseline has mean decoded
ratio \(0.970\), mean headroom \(0.030\), and solved rate \(81.0\%\) on that
subset.

For transparency, the original five-family aggregate contains 1500 paired
rows, with win/tie/loss counts \(829/371/300\), an \(80.0\%\) non-worse rate,
mean decoded-ratio change \(+0.1616\), and solved rates \(28.0\%\) and
\(43.8\%\) for the standard and adaptive methods, respectively. Because this
aggregate includes the unresolved knapsack rows and repeated depths, it is not
treated as the central empirical claim. Appendix~\ref{app:xrot-full} records
the source-of-record table and figures.

The resource difference is substantial. Median adaptive-to-standard runtime
ratios range from approximately \(37\times\) to \(48\times\) by problem, with
an overall median near \(46\times\). The overhead arises from the additional
continuous parameters, numerical differentiation, repeated continuous
optimization, and family screening. The evidence therefore concerns solution
quality under a much larger variational and classical-search budget, not a
runtime-efficient replacement for standard QAOA.

\begin{figure}[H]
  \centering
  \paperfig{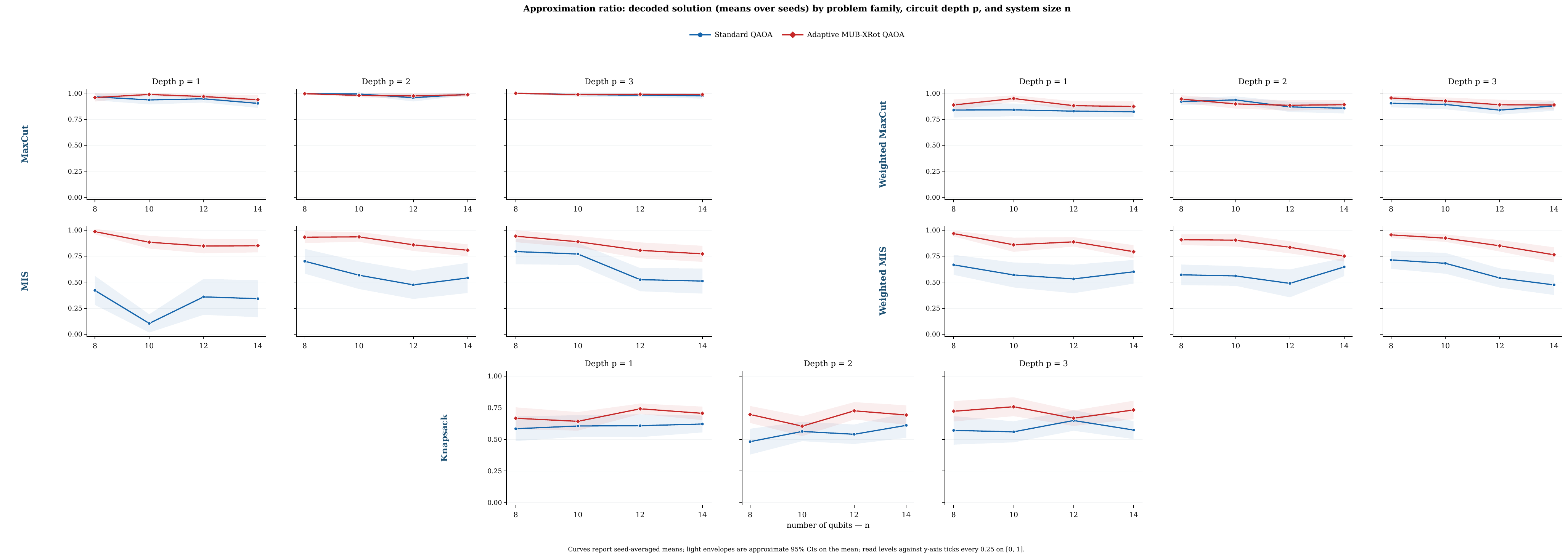}{0.98\textwidth}
  \caption{Decoded solution-ratio facets from the original five-family
  benchmark. The figure is descriptive. The large aggregate change is driven
  primarily by MIS and weighted MIS; the knapsack panel is secondary pending
  resolution of its feasibility/slack-consistency issue.}
  \label{fig:xrot-facets}
\end{figure}

\subsection{Complete benchmark record}

The full source-of-record benchmark contains five problem families, four sizes,
three depths, and 25 seeds as specified in Eq.~\eqref{eq:xrot-grid}. Each method
uses two optimizer restarts. Paired rows are defined at fixed problem, size,
depth, and seed.

\begin{table}[H]
\caption{Original five-family aggregate for adaptive MUB-XRot relative to
standard QAOA. These values are retained without recomputation. The aggregate
is descriptive, includes repeated depths, and includes the secondary knapsack
rows discussed in Sec.~\ref{sec:mub-xrot}.}
\label{tab:xrot-aggregate}
\begin{ruledtabular}
\begin{tabular}{lc}
Quantity & Value \\
\hline
Underlying problem instances & 500 \\
Paired instance--depth rows & 1500 \\
Win/tie/loss & 829/371/300 \\
Non-worse rate & 80.0\% \\
Win rate among non-ties & 73.4\% \\
Mean decoded-ratio change & \(+0.1616\) \\
Solved rate, standard & 28.0\% \\
Solved rate, adaptive & 43.8\% \\
Median runtime ratio & \(\approx46\times\) \\
\end{tabular}
\end{ruledtabular}
\end{table}

\begin{figure}[H]
  \centering
  \paperfig{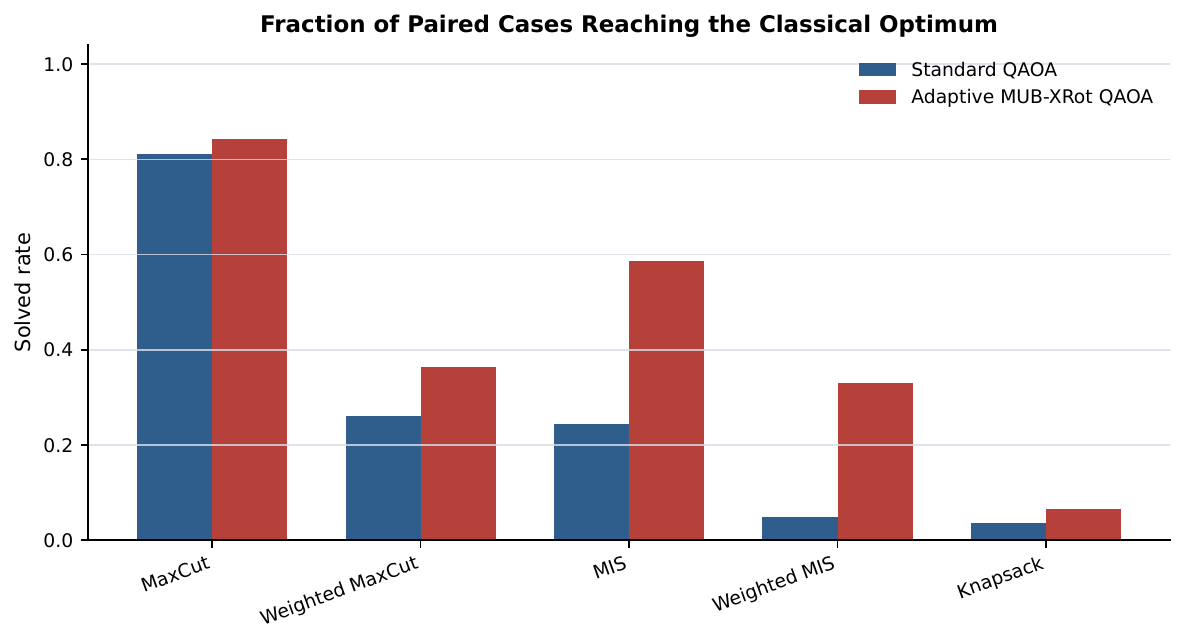}{\columnwidth}
  \caption{Solved rate by problem family in the original benchmark. A row is
  counted as solved when the decoded objective reaches the brute-force optimum.
  The figure does not adjust for repeated depths of the same instance.}
  \label{fig:xrot-solved}
\end{figure}

Candidate families are ranked by a screening objective that includes
postselected expected ratio, decoded ratio, and energy. Since the principal
reported outcome is decoded ratio, screening and evaluation objectives can
disagree. A cluster-aware analysis should treat all depths of an underlying
instance as one resampling cluster. No confidence intervals or \(p\)-values are
reported here because the raw paired rows are not part of the manuscript
source.

\clearpage
\section{Local MUB-family search in non-diagonal QRAO}
\label{app:qrao-full}

This application study concerns MaxCut under a \((3,1)\)-QRAC quantum
relaxation. The relaxed Hamiltonian contains a mixture of Pauli directions and
is generally non-diagonal \cite{AmbainisEtAl2002,TeramotoEtAl2023,FullerEtAl2024QuantumRelaxations}.
Consequently, changing the MUB-family circuit need not commute away as it does
in Proposition~\ref{prop:diagonal-collapse}.

Two objectives must be kept separate. The relaxed QRAO approximation ratio
\(\alpha_r\) evaluates the optimized relaxed Hamiltonian objective. The decoded
classical approximation ratio \(\alpha_c\) evaluates the classical cut obtained
after decoding. An improvement in \(\alpha_r\) is not, by itself, an equal
improvement in \(\alpha_c\).

The local search represents the family index \(r\) as a bit string and probes
valid Hamming-one neighbors
\begin{equation}
  r'=r\oplus2^k.
  \label{eq:qrao-neighbor}
\end{equation}
Its strongest reported variant starts local searches at the two poles
\(r=1\) and \(r=2^{n_q}-1\), retains the better local optimum, and is called
bit-flip multi-start 2POLE. This is a reduced-cost local family search; no
asymptotic bound is proved for the number of family evaluations, and it has no
guarantee of finding the exhaustive best family.

The benchmark uses Erd\H{o}s--R\'{e}nyi \(G(n,0.5)\) MaxCut instances
\cite{ErdosRenyi1960}, with
\begin{equation}
  n\in\{6,8,10,12\},\qquad
  p\in\{1,2,3\},\qquad
  30\text{ seeds}.
  \label{eq:qrao-grid}
\end{equation}
Eight strategies are evaluated on 360 paired graph--depth cells, for 2880
method evaluations. As in the previous study, repeated depths for a graph make
these descriptive paired-cell summaries rather than independent-sample
inference.

\begin{figure}[H]
  \centering
  \paperfig{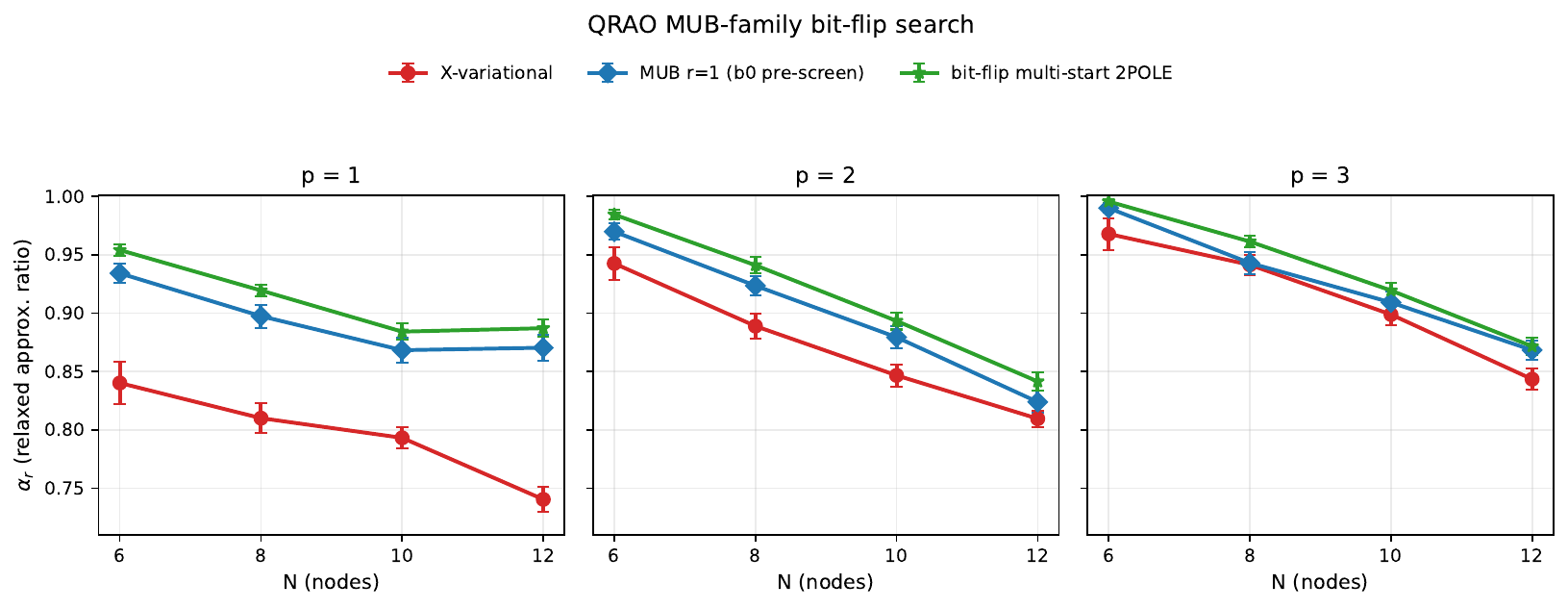}{0.88\textwidth}
  \caption{Local MUB-family search in QRAO MaxCut. The displayed metric is the
  relaxed approximation ratio \(\alpha_r\), not the decoded classical ratio.
  The comparison includes the \(X\)-variational baseline, fixed family \(r=1\)
  with a classical label prescreen, and bit-flip multi-start 2POLE.}
  \label{fig:qrao-main}
\end{figure}

\begin{table}[H]
\caption{Descriptive QRAO comparison on 360 paired graph--depth cells. Solved
rates use the decoded classical ratio \(\alpha_c\), whereas all reported
\(\Delta\alpha_r\) values use the relaxed ratio.}
\label{tab:qrao-main}
\begin{ruledtabular}
\begin{tabular}{lc}
Quantity & Value \\
\hline
Win/tie/loss in \(\alpha_r\) versus \(X\) & 268/20/72 \\
Non-worse rate in \(\alpha_r\) versus \(X\) & 80.0\% \\
Mean \(\Delta\alpha_r\) versus \(X\) & \(+0.0608\) \\
Mean \(\alpha_r\), local multi-start & 0.9211 \\
Mean \(\Delta\alpha_r\) versus fixed \(r=1\) & \(+0.015\) \\
Mean \(\Delta\alpha_r\) versus 2POLE & \(+0.011\) \\
Solved rate from \(\alpha_c\), \(X\) & 46.9\% \\
Solved rate from \(\alpha_c\), local multi-start & 61.9\% \\
\end{tabular}
\end{ruledtabular}
\end{table}

The bit-flip multi-start method has mean \(\alpha_r=0.9211\) and mean gain
\(+0.0608\) over the \(X\)-variational QRAO baseline. It loses to that
baseline in 72 of 360 cells, so the gain is not uniform. The depth dependence
is consistent with a finite-depth basin-finding effect: the reported relaxed
ratio gain is approximately \(+0.09\) to \(+0.15\) at \(p=1\), depending on
size, and narrows to approximately \(+0.02\) to \(+0.03\) at \(p=3\).

The number of family evaluations grows from means \(7.2,10.0,13.8,16.5\) for
\(n=6,8,10,12\), respectively. The corresponding row-averaged exhaustive
counts are \(8.9,14.2,26.2,37.4\). These data show an empirical reduction in
family evaluations over the tested range, not an asymptotic scalability result.
The median runtime ratio relative to the \(X\)-variational QRAO baseline is
approximately \(9.2\times\).

Most of the observed mean gain is not generated by the local family walk. The
reported decomposition is
\begin{equation}
  0.0608=0.0462+0.0039+0.0107,
  \label{eq:qrao-decomposition}
\end{equation}
where \(+0.0462\) comes from the classical \(b_0\) label prescreen,
\(+0.0039\) from comparing the two poles, and \(+0.0107\) from the bit-flip
local search. Thus the experiment supports a composite classical--variational
selection procedure for a structured non-diagonal relaxation. It does not show
that MUB-family circuit dynamics alone produces the full \(+0.0608\) change.

\subsection{Exhaustive-family diagnostic}

An exhaustive sweep over tested QRAO MUB families supplies a diagnostic upper
reference rather than a deployable algorithm. Across 688 graph--depth cases and
13,920 optimization rows, the tested-family oracle improves decoded MaxCut
ratio over the \(X\)-variational baseline by mean \(+0.0811\), with
win/tie/loss counts \(312/373/3\). Relative to a \(Z\)-variational baseline,
the mean decoded-ratio change is \(+0.0552\), with counts \(234/452/2\). A
simple 2POLE heuristic has mean heuristic-to-oracle ratio \(0.9894\). These
numbers motivate the local-search study but do not include the oracle's family
enumeration cost in a practical resource comparison.

\clearpage
\section{Secondary variational checks}
\label{app:secondary}

\subsection{Transverse-field Ising VQE initialization check}

MUB starts were also tested for VQE \cite{PeruzzoEtAl2014,CerezoEtAl2021} on
the open-chain transverse-field Ising Hamiltonian \cite{Pfeuty1970}
\begin{equation}
  H_{\mathrm{TFIM}}=-\sum_i Z_iZ_{i+1}-h\sum_iX_i.
\end{equation}
For exhaustive \(n=5\) sweeps, the best MUB start improves the final energy gap
relative to the best computational-basis start by \(0.9292\), \(0.5215\), and
\(0.4570\) at \(h=0.5,1.0,1.5\), respectively. Sampled \(n=8\) results are
unfavorable. This mixed outcome is retained as a finite-size initialization
check and gives no evidence of favorable scaling.

\subsection{Preliminary non-diagonal matched-MUB-pair experiment}

A matched-MUB-pair prototype was tested for MIS with
\begin{equation}
  H_C^{\mathrm{ND}}
  =\frac{\Omega}{2}\sum_iX_i
   +V\sum_{(i,j)\in E}n_in_j
   -\Delta\sum_i n_i,
  \qquad n_i=\frac{\Id-Z_i}{2},
\end{equation}
a non-diagonal Hamiltonian related to Hamiltonian and Rydberg-array approaches
to MIS \cite{WuYuWilczek2020,EbadiEtAl2022,ZhaoEtAl2025}. On paired
\(n=8,10,12\) Erd\H{o}s--R\'{e}nyi graph--depth cases, adaptive pair selection
has mean decoded-MIS change \(+0.1916\) relative to a fixed non-diagonal
matched baseline, with counts \(127/89/51\). Relative to standard diagonal
QAOA, however, the mean change is only \(+0.0079\), with counts \(94/76/97\),
and the runtime overhead is high. The experiment is preliminary and is not used
to support the main stochastic-extremality theorem or a general variational claim.

\clearpage

\bibliographystyle{unsrtnat}
\bibliography{references}

@article{WoottersFields1989,
  author  = {Wootters, William K. and Fields, Brian D.},
  title   = {Optimal State-Determination by Mutually Unbiased Measurements},
  journal = {Annals of Physics},
  year    = {1989},
  volume  = {191},
  number  = {2},
  pages   = {363--381},
  doi     = {10.1016/0003-4916(89)90322-9}
}

@article{BandyopadhyayEtAl2002,
  author        = {Bandyopadhyay, Somshubhro and Boykin, P. Oscar and Roychowdhury, Vwani and Vatan, Farrokh},
  title         = {A New Proof for the Existence of Mutually Unbiased Bases},
  journal       = {Algorithmica},
  year          = {2002},
  volume        = {34},
  number        = {4},
  pages         = {512--528},
  doi           = {10.1007/s00453-002-0980-7},
  eprint        = {quant-ph/0103162},
  archivePrefix = {arXiv}
}

@inproceedings{KlappeneckerRoetteler2005,
  author        = {Klappenecker, Andreas and R{\"o}tteler, Martin},
  title         = {Mutually Unbiased Bases are Complex Projective 2-Designs},
  booktitle     = {Proceedings of the 2005 IEEE International Symposium on Information Theory},
  year          = {2005},
  pages         = {1740--1744},
  publisher     = {IEEE},
  doi           = {10.1109/ISIT.2005.1523643},
  eprint        = {quant-ph/0502031},
  archivePrefix = {arXiv}
}

@article{RoyScott2007,
  author        = {Roy, Aidan and Scott, A. J.},
  title         = {Weighted Complex Projective 2-Designs from Bases: Optimal State Determination by Orthogonal Measurements},
  journal       = {Journal of Mathematical Physics},
  year          = {2007},
  volume        = {48},
  number        = {7},
  pages         = {072110},
  doi           = {10.1063/1.2759449},
  eprint        = {quant-ph/0703025},
  archivePrefix = {arXiv}
}

@misc{WangWu2024MUBCircuits,
    author        = {Wang, Yu and Wu, Dongsheng},
    title         = {An Efficient Quantum Circuit Construction Method for Mutually Unbiased Bases in \(n\)-Qubit Systems},
    year          = {2024},
    eprint        = {2311.11698},
    archivePrefix = {arXiv},
    primaryClass  = {quant-ph},
    doi           = {10.48550/arXiv.2311.11698},
  note          = {arXiv preprint arXiv:2311.11698},
  url           = {https://doi.org/10.48550/arXiv.2311.11698}
}

@article{PeruzzoEtAl2014,
  author  = {Peruzzo, Alberto and McClean, Jarrod and Shadbolt, Peter and Yung, Man-Hong and Zhou, Xiao-Qi and Love, Peter J. and Aspuru-Guzik, Al{\'a}n and O'Brien, Jeremy L.},
  title   = {A Variational Eigenvalue Solver on a Photonic Quantum Processor},
  journal = {Nature Communications},
  year    = {2014},
  volume  = {5},
  pages   = {4213},
  doi     = {10.1038/ncomms5213}
}

@article{CerezoEtAl2021,
  author  = {Cerezo, M. and Arrasmith, Andrew and Babbush, Ryan and Benjamin, Simon C. and Endo, Suguru and Fujii, Keisuke and McClean, Jarrod R. and Mitarai, Kosuke and Yuan, Xiao and Cincio, Lukasz and Coles, Patrick J.},
  title   = {Variational Quantum Algorithms},
  journal = {Nature Reviews Physics},
  year    = {2021},
  volume  = {3},
  pages   = {625--644},
  doi     = {10.1038/s42254-021-00348-9}
}

@misc{TeramotoEtAl2023,
  author        = {Teramoto, Kosei and Raymond, Rudy and Wakakuwa, Eyuri and Imai, Hiroshi},
  title         = {Quantum-Relaxation Based Optimization Algorithms: Theoretical Extensions},
  year          = {2023},
  eprint        = {2302.09481},
  archivePrefix = {arXiv},
  primaryClass  = {quant-ph},
  doi           = {10.48550/arXiv.2302.09481},
  note          = {arXiv preprint arXiv:2302.09481},
  url           = {https://doi.org/10.48550/arXiv.2302.09481}
}

@article{Pfeuty1970,
  author  = {Pfeuty, Pierre},
  title   = {The One-Dimensional {Ising} Model with a Transverse Field},
  journal = {Annals of Physics},
  year    = {1970},
  volume  = {57},
  number  = {1},
  pages   = {79--90},
  doi     = {10.1016/0003-4916(70)90270-8}
}

@article{EbadiEtAl2022,
  author  = {Ebadi, Sepehr and Keesling, Alexander and Cain, Madelyn and Wang, Tout T. and Levine, Harry and Bluvstein, Dolev and Semeghini, Giulia and Omran, Ahmed and Liu, Jinguo and Samajdar, Rhine and Pichler, Hannes and Choi, Soonwon and Irani, Sandy and Greiner, Markus and Vuleti{\'c}, Vladan and Lukin, Mikhail D.},
  title   = {Quantum Optimization of Maximum Independent Set Using {Rydberg} Atom Arrays},
  journal = {Science},
  year    = {2022},
  volume  = {376},
  number  = {6598},
  pages   = {1209--1215},
  doi     = {10.1126/science.abo6587}
}

@article{ZhaoEtAl2025,
  author  = {Zhao, Xianjue and Ge, Peiyun and Yu, Hongye and You, Li and Wilczek, Frank and Wu, Biao},
  title   = {Quantum Hamiltonian Algorithms for Maximum Independent Sets},
  journal = {National Science Review},
  year    = {2025},
  volume  = {12},
  number  = {9},
  pages   = {nwaf304},
  doi     = {10.1093/nsr/nwaf304},
  eprint  = {2310.14546},
  archivePrefix = {arXiv},
  primaryClass = {quant-ph}
}

@article{WuYuWilczek2020,
  author  = {Wu, Biao and Yu, Hongye and Wilczek, Frank},
  title   = {Quantum Independent-Set Problem and {Non-Abelian} Adiabatic Mixing},
  journal = {Physical Review A},
  year    = {2020},
  volume  = {101},
  pages   = {012318},
  doi     = {10.1103/PhysRevA.101.012318}
}

@book{BoucheronLugosiMassart2013,
  author    = {Boucheron, St{\'e}phane and Lugosi, G{\'a}bor and Massart, Pascal},
  title     = {Concentration Inequalities: A Nonasymptotic Theory of Independence},
  publisher = {Oxford University Press},
  year      = {2013},
  isbn      = {9780199535255}
}

@book{LedouxTalagrand1991,
  author    = {Ledoux, Michel and Talagrand, Michel},
  title     = {Probability in Banach Spaces: Isoperimetry and Processes},
  publisher = {Springer},
  year      = {1991},
  series    = {Ergebnisse der Mathematik und ihrer Grenzgebiete},
  volume    = {23},
  isbn      = {9783540520139}
}

@article{ByrdLuNocedalZhu1995,
  author  = {Byrd, Richard H. and Lu, Peihuang and Nocedal, Jorge and Zhu, Ciyou},
  title   = {A Limited Memory Algorithm for Bound Constrained Optimization},
  journal = {SIAM Journal on Scientific Computing},
  year    = {1995},
  volume  = {16},
  number  = {5},
  pages   = {1190--1208},
  doi     = {10.1137/0916069}
}

@misc{AlfassiMeiromMor2024DQES,
    author        = {Alfassi, Ittay and Meirom, Dekel and Mor, Tal},
    title         = {Discretized Quantum Exhaustive Search for Variational Quantum Algorithms},
    year          = {2024},
    eprint        = {2407.17659},
    archivePrefix = {arXiv},
    primaryClass  = {quant-ph},
    doi           = {10.48550/arXiv.2407.17659},
  note          = {arXiv preprint arXiv:2407.17659},
  url           = {https://doi.org/10.48550/arXiv.2407.17659}
}

@misc{NakamuraTsuji2025CenteredGCI,
  author        = {Nakamura, Shohei and Tsuji, Hiroshi},
  title         = {The {Gaussian} Correlation Inequality for Centered Convex Sets and the Case of Equality},
  year          = {2025},
  eprint        = {2504.04337},
  archivePrefix = {arXiv},
  primaryClass  = {math.FA},
  doi           = {10.48550/arXiv.2504.04337},
  note          = {arXiv preprint arXiv:2504.04337},
  url           = {https://doi.org/10.48550/arXiv.2504.04337}
}

@article{AmbainisEtAl2002,
  author        = {Ambainis, Andris and Nayak, Ashwin and Ta-Shma, Amnon and Vazirani, Umesh},
  title         = {Dense Quantum Coding and Quantum Finite Automata},
  journal       = {Journal of the ACM},
  year          = {2002},
  volume        = {49},
  number        = {4},
  pages         = {496--511},
  doi           = {10.1145/581771.581773},
  eprint        = {quant-ph/9804043},
  archivePrefix = {arXiv}
}

@article{FullerEtAl2024QuantumRelaxations,
  author        = {Fuller, Bryce and Hadfield, Charles and Glick, Jennifer R. and Imamichi, Takashi and Itoko, Toshinari and Thompson, Richard J. and Jiao, Yang and Kagele, Marna M. and Blom-Schieber, Adriana W. and Raymond, Rudy and Mezzacapo, Antonio},
  title         = {Approximate Solutions of Combinatorial Problems via Quantum Relaxations},
  journal       = {IEEE Transactions on Quantum Engineering},
  year          = {2024},
  volume        = {5},
  pages         = {1--15},
  doi           = {10.1109/TQE.2024.3421294},
  eprint        = {2111.03167},
  archivePrefix = {arXiv},
  primaryClass  = {quant-ph}
}

@article{ErdosRenyi1960,
  author  = {Erd{\H{o}}s, Paul and R{\'e}nyi, Alfr{\'e}d},
  title   = {On the Evolution of Random Graphs},
  journal = {Publications of the Mathematical Institute of the Hungarian Academy of Sciences},
  year    = {1960},
  volume  = {5},
  pages   = {17--61}
}

@article{Linhart1979EdgeCurvature,
  author  = {Linhart, J.},
  title   = {{Kantenkr{\"u}mmung und Umkugelradius konvexer Polyeder}},
  journal = {Acta Mathematica Academiae Scientiarum Hungaricae},
  year    = {1979},
  volume  = {34},
  pages   = {1--2},
  doi     = {10.1007/BF01902585}
}

@book{Schneider2014ConvexBodies,
  author    = {Schneider, Rolf},
  title     = {Convex Bodies: The Brunn--Minkowski Theory},
  publisher = {Cambridge University Press},
  year      = {2014},
  edition   = {2},
  series    = {Encyclopedia of Mathematics and its Applications},
  volume    = {151},
  doi       = {10.1017/CBO9781139003858}
}

@article{SpenglerEtAl2012MUBEntanglement,
  author  = {Spengler, Christoph and Huber, Marcus and Brierley, Stephen and Adaktylos, Theodor and Hiesmayr, Beatrix C.},
  title   = {Entanglement detection via mutually unbiased bases},
  journal = {Physical Review A},
  volume  = {86},
  pages   = {022311},
  year    = {2012},
  doi     = {10.1103/PhysRevA.86.022311}
}

@misc{SemreFrankel2026RandSEMI,
  author        = {Semre, Emilio and Frankel, Steven},
  title         = {{Rand-SEMI-QAOA}: Finite-Budget Depth-One {MaxCut} Ensembles on Compressed Quantum Registers},
  year          = {2026},
  eprint        = {2608.15655},
  archivePrefix = {arXiv},
  primaryClass  = {quant-ph},
  note          = {arXiv preprint arXiv:2608.15655},
  url           = {https://arxiv.org/abs/2608.15655}
}

@article{McCleanEtAl2018BarrenPlateaus,
  author  = {McClean, Jarrod R. and Boixo, Sergio and Smelyanskiy, Vadim N. and Babbush, Ryan and Neven, Hartmut},
  title   = {Barren plateaus in quantum neural network training landscapes},
  journal = {Nature Communications},
  volume  = {9},
  pages   = {4812},
  year    = {2018},
  doi     = {10.1038/s41467-018-07090-4}
}

@article{CerezoEtAl2021CostDependent,
  author  = {Cerezo, M. and Sone, Akira and Volkoff, Tyler and Cincio, Lukasz and Coles, Patrick J.},
  title   = {Cost function dependent barren plateaus in shallow parametrized quantum circuits},
  journal = {Nature Communications},
  volume  = {12},
  pages   = {1791},
  year    = {2021},
  doi     = {10.1038/s41467-021-21728-w}
}

\end{document}